\documentclass[12pt]{article}

\usepackage{amsmath}
\usepackage{amsfonts,graphicx}
\usepackage{amssymb,amsthm}
\usepackage{subfigure}
\usepackage{mathrsfs}

\def\R{\mathbb{R}}
\def\N{\mathbb{N}}

\theoremstyle{plain}
\newtheorem{theorem}{Theorem}
\newtheorem{Prop}{Proposition}

\theoremstyle{remark}

\begin{document}

\title{\textbf{Chemical Oscillations out of Chemical Noise}}
\author{Carlos Escudero\\
\small \textsl{Departamento de Econom\'{\i}a Cuantitativa} \\
\small \textsl{\& Instituto de Ciencias Matem\'aticas (CSIC-UAM-UC3M-UCM)}\\
\small \textsl{Universidad Aut\'{o}noma de
Madrid, 28049 Madrid, Spain}\\
\small \textit{e-mail: cel@icmat.es}\\
Andr\'es M. Rivera\\
\small \textsl{Departamento de Ciencias Naturales y Matem\'aticas, Facultad de Ingenier\'ia}\\
\small \textsl{Pontificia Universidad Javeriana Cali, 26239 Cali, Colombia}\\
\small \textit{e-mail: amrivera@puj.edu.co}\\
Pedro J. Torres\\
\small \textsl{Departamento de Matem\'atica Aplicada, Facultad de Ciencias }\\
\small \textsl{Universidad de Granada, 18071 Granada, Spain}\\
\small \textit{e-mail: ptorres@ugr.es}}
\date{}
\maketitle

\begin{abstract}
The dynamics of one species chemical kinetics is studied. Chemical
reactions are modelled by means of continuous time Markov
processes whose probability distribution obeys a suitable master
equation. A large deviation theory is formally introduced, which
allows developing a Hamiltonian dynamical system able to describe
the system dynamics. Using this technique we are able to show that
the intrinsic fluctuations, originated in the discrete character
of the reagents, may sustain oscillations and chaotic trajectories
which are impossible when these fluctuations are disregarded. An
important point is that oscillations and chaos appear in systems
whose mean-field dynamics has too low a dimensionality for showing
such a behavior. In this sense these phenomena are purely induced
by noise, which does not limit itself to shifting a bifurcation
threshold. On the other hand, they are large deviations of a short
transient nature which typically only appear after long waiting
times. We also discuss the implications of our results in
understanding extinction events in population dynamics models
expressed by means of stoichiometric relations.
\end{abstract}

\maketitle

\section{Introduction}

The kinetics of reaction systems have been extensively studied
over the years. These systems have been postulated as paradigmatic
models for the description of a large number of natural phenomena,
including topics from organic and inorganic chemistry,
epidemiology, population biology and gene\-tics, nuclear physics,
non-equilibrium statistical mechanics and many others
sciences~\cite{kampen,gardiner}. The mathematical description of
such systems usually starts with the assumption of a set of
stoichiometric relations of the form
\begin{equation}
A \stackrel{\alpha}\longrightarrow B,
\end{equation}
signifying that the reagent $A$ transforms to $B$ with the time
dependent proba\-bility
\begin{equation}
P_{A \to B}(t)=\alpha e^{-\alpha t},
\end{equation}
where $\alpha >0$ is the specific reaction rate. Note that the
probabilistic nature of the reactions introduces fluctuations into
the dynamics: this is the ``chemical noise'' we will be interested
in. In more general terms, the state of a system containing $m$
reagents and $n$ reactions is described by a continuous time
Markov process. All the available information is carried by the
distribution $P(n_A,n_B,\cdots;t)$ specifying the probability of
the existence of exactly $n_A$ reagents of type $A$, $n_B$ of type
$B$, ..., at time $t$. The dynamics of this distribution is
governed by a master equation of the form~\cite{kampen,gardiner}
\begin{equation}
\frac{dP_k}{dt}= \sum_{j \neq k} \left[ W_{j \to k} P_j - W_{k \to
j} P_k \right],
\end{equation}
where $P_k$ denotes the probability of finding the system in the
state $k$ charac\-terized by a certain number of reagents of each
type, and $W_{j \to k}$ is the transition matrix from state $j$ to
state $k$. While solving the master equation to find $P_k$ would
mean that we possess all the available information on the system,
the chances of obtaining an exact closed form for $P_k$ are scarce
in realistic situations. Additionally, the amount of information
it brings is usually excessive, and a great part of it does not
add any valuable information about the dynamics. Consequently, the
most common approach to the subject concentrates on the dynamics
of some statistical quantity of interest, as for instance a
density, which is able to describe the system macroscopic state.
The selection of an adequate variable has to be supplemented with
selecting a suitable approximation in order to get an operative
theory that allows studying the otherwise commonly untractable
master equation. A particularly advantageous choice is the analog
of the quantum mechanical Wentzel-Kramers-Brillouin (WKB)
approximation adapted to this sort of systems, which is now well
established in both physical and mathematical literatures, see for
instance~\cite{kubo,schuss1,schuss2,schuss3,dykman1,dykman2,elgart1,
doering1,elgart2,doering2,assaf1,assaf2,escudero}. It allows the
description of both the short time dynamics, which is to a large
extent independent of the fluctuations and therefore captured by
mean-field type approximations, and the long time behavior which
is affected, dramatically on occasion, by large deviations. The
mathematical and physical natures of these large deviations will
appear clearer in the following sections.

In this work we are concerned with simple reaction sets which on
the other hand have an intuitive physical meaning. This way we
will focus on simplified model systems which, although not of
direct practical applicability, facilitate analytic progress and
physical intuition. We will show how chemical fluctuations
strongly affect the dynamics for long times, when large deviations
have had time to develop. In these cases, chemical fluctuations
are able to promote periodic orbits and chaotic behavior in
reaction systems whose dimensionality is too low to present such a
behavior if strictly deterministic dynamics are considered. These
effects are, however, both rare and short-lived. They are, at the
same time, substantially different from other sorts of
noise-induced oscillations which appear in different important
phenomena and whose structure relies on an existing deterministic
mechanism (like the proximity to a deterministic bifurcation)
which is anticipated or activated by
noise~\cite{TS1985,POC1999,Lindner-etal2004,Liebermeister2005,DR2006}.
In this sense, we may say our focus is on oscillations which are
purely promoted by chemical noise. Our approach will be
probabilistic at the beginning, when we will present formal
calculations in which the equations governing large fluctuations
will be derived. These equations have the form of Hamiltonian
dynamical systems, which will be in realistic situations genuinely
different from the ones usually considered in classical mechanics.
For them we will be able to show, by means of explicit
calculations, rigorous proofs and numerical simulations, how
chemical noise is capable of sustaining chemical oscillations,
both of periodic and chaotic nature.

\section{Large deviations}

\subsection{Brownian motion}

We devote this section to clarify the type of problems we are
going to solve. We start with the perhaps simplest stochastic
process one could consider: one-dimensional Brownian motion. For
our current purposes it will be the solution of the equation
\begin{equation}
\label{brownian} \frac{dB}{dt}=D \, \xi(t),
\end{equation}
where $D>0$ is a diffusion constant and $\xi(t)$ is the standard
Gaussian white noise defined by its two first moments
\begin{equation}
\mathbb{E}[\xi(t)]=0, \qquad \mathbb{E}[\xi(t) \xi(t')]=
\delta(t-t'),
\end{equation}
where $\delta(\cdot)$ is the Dirac delta distribution. Of course,
a rigorous interpretation of this equation is possible in terms of
It\^o calculus~\cite{oksendal}, but such a precise definition will
not be needed herein. Equation~(\ref{brownian}) is provided with
the initial condition $B(0)=0$. A classical problem within this
subject is calculating the first time the random walker $B(t)$
reaches some fixed level $a \neq 0$ in absence of other
constraints. The well--known solution states that the random
walker reaches level $a$ in finite time with probability one, but
the mean time at which this event occurs diverges.

Langevin equations like~(\ref{brownian}) and more complicated
variants of it are well understood from the large deviations point
of view~\cite{freidlin}. It associates to this stochastic
differential equation the rate or action functional
\begin{equation}
S[x(t)]= \frac{1}{2D} \int |\dot{x}(t')|^2 dt',
\end{equation}
which in the small noise limit $D \to 0$ gives rise to the
following Euler-Lagrange equation
\begin{equation}
\ddot{x}=0,
\end{equation}
for the position $x$ of the random walker. If we complement this
equation with the boundary conditions $x(0)=0$ and $x(T)=a$ we
find the solution
\begin{equation}
\label{largedeviation}
x(t)=\frac{a}{T}t,
\end{equation}
signaling the most probable trajectory that links the origin with
the level $a$ after a time $T$ for the Brownian
dynamics~(\ref{brownian}). Small deviations are obtained by
setting $D=0$ in~(\ref{brownian}). In this case the random walker
stays at the origin for all times. So the full picture would be,
for small noise, the random walker will be at a neighborhood of
the origin with a large probability but with a small probability
large deviations might appear and drive the system further away.
The probability $\mathcal{P}$ with which these large deviations,
which promote trajectories~(\ref{largedeviation}), manifest
themselves into the system dynamics is proportional to the
exponential of the negative of the action
\begin{equation}
\label{gaussian} \mathcal{P} \sim e^{-S} = \exp \left(-
\frac{a^2}{2DT} \right).
\end{equation}
From this formula it is clear that those trajectories that connect
the origin with the level $a$ in a shorter time are rarer than
those which take a longer time. The prefactor in this case is
easily found by normalization. Note that the large deviation
theory is valid for $a^2 \gg DT$, otherwise the system diffuses
away from the original position and the approximation breaks down.

\subsection{Plankton extinction}

We will describe now the large deviations technique in the context
of reaction kinetics. To this end we consider a simple model that
has nevertheless a genuine practical interest. This model was
introduced as a description of plankton population
dynamics~\cite{zhang,adler,young} and nonequilibrium statistical
mechanics~\cite{henkel}. It consists of the following reactions
\begin{equation}
A \stackrel{\gamma}\longrightarrow 2A, \qquad A
\stackrel{\gamma}\longrightarrow \emptyset,
\end{equation}
happening at the same rate $\gamma$. We will employ large
deviation theory to describe the extinction of the ``plankton
population'' $A$. The state of the system may be described by a
continuous time Markov process obeying the master equation
\begin{equation}
\frac{d P_n}{dt}=\gamma [(n-1)P_{n-1}-nP_n] + \gamma
[(n+1)P_{n+1}-nP_n],
\end{equation}
where the term inside the first bracket corresponds to the
branching reaction and the one inside the second bracket to the
disintegration reaction. By introducing the generating function
\begin{equation}
G(p,t)= \sum_{n=0}^\infty p^n P_n(t),
\end{equation}
we transform the master equation into the following ``imaginary
time Schr\"{o}dinger equation''
\begin{equation}
\partial_t G = \gamma (p-1)^2 \partial_p G.
\end{equation}
Note that we are employing the ``momentum'' rather than the
``coordinate'' representation in this last equation. One can
obtain the probability distribution from the generating function
in the following way
\begin{equation}
\label{contint}
P_n(t)= \frac{1}{2 \pi i} \oint G(p,t) \, p^{-n} \, \frac{dp}{p},
\end{equation}
where the contour integral runs over a closed path on the complex
$p-$plane, surrounding the origin and going through the region of
analyticity of $G(p,t)$. The corresponding ``classical''
Hamiltonian of our theory reads
\begin{equation}
\label{classicalhamiltonian}
\mathcal{H}=\gamma (p-1)^2 q;
\end{equation}
it is precisely this Hamiltonian, as in the previous case, which
describes the large deviations of the system. The corresponding
equations of motion are
\begin{eqnarray}
\label{dotq}
\dot{q} &=& \frac{\partial \mathcal{H}}{\partial p}= 2 \gamma
(p-1) q,
\\
\dot{p} &=& -\frac{\partial \mathcal{H}}{\partial q}= -\gamma
(p-1)^2.
\label{dotp}
\end{eqnarray}
Note that the line $p=1$ is degenerate and all points on it are
fixed points. These solutions refer to small deviations: the
system stays with a large probability in a neighborhood of the
initial condition for short times as in the Brownian motion case.
The solution for the coordinate $q$ is
\begin{equation}
\label{solucionq}
q(t)= \left( \sqrt{q(0)} \pm \sqrt{\gamma H}t
\right)^2,
\end{equation}
where the minus sign is selected for extinction trajectories and
$H$ is a constant indicating the initial ``energy''. The duplicity
of signs in this equation comes from the extraction of the square
root of equation~(\ref{classicalhamiltonian}). The number $n$ of
reagents can be calculated by means of formally applying a
steepest descent approximation to formula~(\ref{contint}) for
$G(p,t)=\exp[-S(p,t)]$, where $S$ is the ``classical''
action~\cite{elgart1}. Then one finds $n \approx -p \partial_p S$,
and employing the ``classical'' relation $q = -\partial_p S$ one
concludes $n(t) \approx p(t) q(t)$. In our particular example one
finds
\begin{equation} \label{numero}
n \approx pq= q -\sqrt{\frac{Hq}{\gamma}},
\end{equation}
which becomes zero due to a fluctuation when $p \to 0$, which
leads to $q=H/\gamma$. If the system follows an optimal trajectory
it will become extinct after a time
\begin{equation}
t_e=\sqrt{\frac{q(0)}{\gamma H}}-\frac{1}{\gamma}.
\end{equation}
The probability with which this realization appears for short
times is the exponentiated negative of the action $\mathcal{P}
\sim e^{-S}$, up to some prefactor. We will limit ourselves to the
exponential order, as the calculation of the prefactor is a rather
technical issue~\cite{escudero} and will not add substantial
information to the present discussion. The action reads
\begin{equation}
S= \int_0^{t_e} (p\dot{q}-H) dt + [p(0)q(0)-p(t_e)q(t_e)]
+S_0=S_0,
\end{equation}
in the case of an extinction trajectory, where the initial action
$S_0=- \ln [G(p,0)]$. The last equality has been derived employing
the following derivations
\begin{eqnarray}
\int_0^{t_e} (p\dot{q}-H) dt &=& \int_0^{t_e} [p\dot{q}- \gamma
(p-1)^2 q] dt = \int_0^{t_e} \gamma q (p^2-1) dt, \\
\frac{d}{dt} (pq) &=& \gamma q (p^2-1),
\end{eqnarray}
where we have substituted $\dot{p}$ and $\dot{q}$ for their
respective values from (\ref{dotq}) and (\ref{dotp}). We consider
two initial conditions as in \cite{elgart1}, the Poisson
distributed initial condition with average $n_0$, this is
$G(p,0)=\exp [n_0(p-1)]$, and the Kronecker delta centered at
$n_0$, which is $G(p,0)=p^{n_0}$. In the first case the extinction
probability reads
\begin{equation}
\mathcal{P} \sim \exp \left[ -\sqrt{\frac{H^2}{4 \gamma^2} +
\frac{H n_0}{\gamma}}+\frac{H}{2 \gamma} \right],
\end{equation}
and in the second
\begin{equation}
\mathcal{P} \sim \left[ 1+ \frac{H}{2 n_0 \gamma}
-\sqrt{\frac{H^2}{4 n_0^2 \gamma^2} + \frac{H}{n_0 \gamma}}
\right]^{n_0},
\end{equation}
and both yield the same result in the thermodynamic limit $n_0 \to
\infty$
\begin{equation}
\label{ctime} \mathcal{P} \sim \exp \left( -\sqrt{\frac{H
n_0}{\gamma}} \right).
\end{equation}
The optimal trajectory corresponding to this characteristic time
to extinction is found by combining Eq. (\ref{solucionq}) (with
the minus sign) and the expression for $n$ given by the second
equality of Eq. (\ref{numero})
\begin{equation}
\label{optimal} n(t)= n_0 + \gamma H t^2 -t \sqrt{H^2 + 4 n_0
\gamma H} \approx \left( \sqrt{n_0} - \sqrt{\gamma H} \, t
\right)^2,
\end{equation}
where the thermodynamic limit $n_0 \to \infty$ has been considered
in the last step. In the derivation of the first equality we have
employed the following two relations
\begin{eqnarray}
n(t) &=& \left( n_0 -\sqrt{\gamma H q(0)} \, t \right)
\left(1-\sqrt{\frac{\gamma
H}{q(0)}} \, t \right), \\
\sqrt{q(0)}+ \frac{n_0}{\sqrt{q(0)}} &=& \sqrt{\frac{H}{\gamma}+4
\, n_0}.
\end{eqnarray}
Note that we have found a one parameter family of solutions,
parameterized with the energy $H$. Time $t_e$ is not the mean
extinction time, because at every time there are equally probable
trajectories which do not become extinct, the ones corresponding
to the plus sign in Eq. (\ref{solucionq}). This makes the mean
extinction time infinite, although the system becomes extinct with
probability one~\cite{gardiner,escudero2}. The interpretation of
this time is that of a characteristic time of extinction, this is,
if the system becomes extinct after a time $t_e$ then the most
probable path to extinction would have been (\ref{optimal}). Note
that more ``energetic'' trajectories lead to extinction faster but
they are less probable. Using the relation between $t_e$ and $H$
we may find an expression akin to~(\ref{gaussian}):
\begin{equation}
\mathcal{P} \sim \exp \left( -\frac{n_0}{\gamma t_e} \right).
\end{equation}
Note that in this case the scaling is different. Note also that,
as in the previous case, the large deviation theory that has led
us to the optimal trajectories~(\ref{optimal}) is valid for short
times $t \ll n_0/\gamma$.

\section{Chemical oscillations and chaos}

\subsection{Noise induced oscillations}

\begin{figure}[h]
\begin{center}
\includegraphics[scale=0.7]{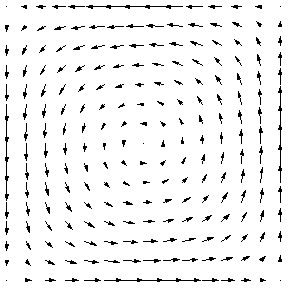}
\caption{Vector field of the Hamiltonian system \eqref{ham}
(detail of the qua\-dran\-gular area enclosed by the four zero
energy lines). The values of the parameters are $\mu=1$ and
$\sigma=2$.} \label{quadarea}
\end{center}
\end{figure}

We now move to studying the more complex nonlinear reversible
reaction
\begin{equation}
\label{reactionset}
A \stackrel{\mu}\longrightarrow 2A, \qquad 2A
\stackrel{\sigma}\longrightarrow A.
\end{equation}
It is clear that it can be considered as a stochastic discrete
model for logistic growth. The master equation describing this
reactions set is
\begin{equation}
\frac{dP_n}{dt}=\mu [(n-1)P_{n-1}-nP_n] + \frac{\sigma}{2}
[(n+1)nP_{n+1}-n(n-1) P_n].
\end{equation}
We may use the generating function technique to convert this
equation into a partial differential equation which can be exactly
mapped, using quantum mechanical tools, into the path integral
\cite{elgart1}
\begin{equation}
U=\int \mathcal{D}p\mathcal{D}qe^{-S[p,q]},
\end{equation}
where $U$ is the problem Green function. The action reads
\begin{equation}
S[p,q]=\int_0^t \left[p\dot{q}
-\mu(p^2-p)q-\frac{\sigma}{2}(p-p^2)q^2 \right]dt +
[p(0)q(0)-p(t)q(t)] + S_0,
\end{equation}
and upon rendering the variables nondimensional $q \to
(2\mu/\sigma)q$ (so this new $q=O(1)$, see below) and $t \to
t/\mu$ one finds
\begin{equation}
S[p,q]=\frac{2\mu}{\sigma} \left\{ \int_0^{t/\mu} \left[p\dot{q}
-(p^2-p)q-(p-p^2)q^2 \right]dt + [p(0)q(0)-p(t)q(t)] + \tilde{S}_0
\right\},
\end{equation}
where $\tilde{S}_0=\sigma S_0/(2 \mu)=O(1)$, so the steepest
descent method makes sense for $\mu \gg \sigma$. In this
approximation and recovering the dimensional coordinates the large
deviations problem reduces to studying the Hamiltonian
\cite{elgart2}
\begin{equation}\label{ham}
\mathcal{H} = \mu(p^2-p)q+\frac{\sigma}{2}(p-p^2)q^2,
\end{equation}
and the corresponding dynamical system
\begin{equation}\label{hamilton}
\left\{
\begin{aligned}
\dot{p} &=& -\frac{\partial \mathcal{H}}{\partial q}
=\mu(p-p^2)+\sigma(p^2-p)q,
\\
\dot{q} &=& \frac{\partial \mathcal{H}}{\partial p} =
\mu(2p-1)q+\frac{\sigma}{2}(1-2p)q^2.
\end{aligned}
\right.
\end{equation}
This system has five fixed points, four of which lie in the energy
$H=0$ level, these are
\begin{equation}
(0,0), \quad (0,2\mu/\sigma), \quad (1,0), \quad \mathrm{and}
\quad (1,2\mu/\sigma),
\end{equation}
all of them are saddles. The fifth fixed point is
$(1/2,\mu/\sigma)$, its energy is $H=-\mu^2/(8\sigma)$, and it is
a local minimum of energy, what implies it is a center. The $H=0$
level is particularly simple, as it is composed of the four
invariant lines
\begin{equation}
\{p=0\}, \quad \{p=1\}, \quad \{q=0\}, \quad \mathrm{and} \quad
\{q=2\mu/\sigma\}.
\end{equation}
The dynamics is exactly integrable in all these lines. They cross
at the four zero energy fixed points, and they enclose a
quadrangular area in whose center lies the fifth fixed point. In
this quadrangular area all the trajectories are periodic orbits
surrounding the center, see Fig.~\ref{quadarea}. Note that, both
outside and inside these four zero energy lines the energy is
strictly negative. The mean-field behavior corresponds to the
$\{p=1\}$ line, on which the dynamics reduces to the well known
logistic equation
\begin{equation}
\dot{q} = \mu q - \frac{\sigma}{2} q^2.
\end{equation}
As expected, reaction~(\ref{reactionset}) corresponds to pure
logistic growth when the fluctuations are neglected, this is, for
a large number of reagents and short times. So at the mean-field
level the only possibility is a monotonic approach to the stable
fixed point $q=2\mu/\sigma$, provided the initial condition
fulfills $q>0$. So the dynamical scenario presents a very reduced
phenomenology in this case. However, if we consider the intrinsic
fluctuations and thus the full phase space things are different.
In this case, for instance, we may observe the system for a short
time in the neighborhood of the fifth fixed point, which
represents a reagent density $n \approx \mu/(2\sigma)$ (this is
obtained as the product $pq$ evaluated at the fixed point). If we
initialize the system with this reagent number we have a
probability
\begin{equation}
\mathcal{P} \sim e^{-S[p,q]}= \exp \left( -\frac{\mu^2
\tau}{8\sigma} - \frac{\mu}{4\sigma} \right), \qquad \mathcal{P}
\sim \exp \left[ -\frac{\mu^2 \tau}{8\sigma} - \frac{\mu
\ln(2)}{2\sigma} \right],
\end{equation}
of observing the system in a neighborhood of this point during a
time $\tau$ respectively for the Poissonian distributed and
deterministic initial condition. But furthermore we can observe
periodic behavior. All the periodic orbits are optimal
trajectories that can be observed experimentally if we wait long
enough. These orbits are characterized by an energy
$0>H_p>-\mu^2/(8\sigma)$, and therefore the probability of
observing a number $m$ of cycles is
\begin{equation}
\label{smallprob}
\mathcal{P} \sim e^{H_p m t_p - m \mathcal{A} - S_0},
\end{equation}
again at exponential order, where $t_p$ is the time it takes to
perform one such cycle and $\mathcal{A}$ is the phase space area
enclosed by one such trajectory. We expect that formulas like this
will be able to express the order of magnitude of the
corresponding probability, not just an exponential dependence, as
we have observed in other cases when the system is not in the
proximity of an absorbing state~\cite{escudero}. So we see that,
while the mean-field description predicts monotonic approach to a
stable fixed point, the stochastic theory allows the appearance of
transient periodic trajectories. Let us emphasize that such
trajectories are not the typical behavior of the solution to the
master equation. They are large deviations, which manifest
themselves only after very long times and are of a short transient
nature. These characteristics are quantitatively described by the
small probability of its occurrence~(\ref{smallprob}). Of course,
periodic orbits in the $(p,q)-$plane are not of physical nature.
But it is on the other hand immediate that the number of reagents
$n(t)=p(t)q(t)$ is periodic if both $p(t)$ and $q(t)$ are
periodic. We have represented the time evolution of $n$ for an
initial condition belonging to one of the periodic orbits in the
$(p,q)-$plane in Fig.~\ref{npqt}.

\begin{figure}[h]
\begin{center}
\includegraphics[scale=0.7]{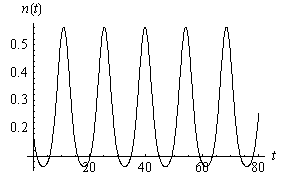}
\caption{Number of reagents $n(t)=p(t)q(t)$ versus time $t$
obtained from numerically integrating dynamical
system~(\ref{hamilton}). The system is initialized with the
conditions $p(0)=1/4$ and $q(0)=3/4$ which correspond to
$n(0)=3/16$. The values of the parameters are $\mu=1$ and
$\sigma=2$.} \label{npqt}
\end{center}
\end{figure}

The fact that Hamiltonians like~(\ref{ham}), which come from a
chemical master equation, are not hermitian translates into the
impossibility of expressing probabilities like~(\ref{smallprob})
in terms of the physical variable $n$ rather than the formal
auxiliary variables $p$ and $q$. Despite this undesirable fact, we
can still characterize periodic orbits like the one represented in
figure~\ref{npqt} by means of its period. Indeed, it is clear that
for periodic solutions, the period of $p(t)$ and $q(t)$ will
uniquely determine the period of $n(t)$. So a way to connect the
physically measurable quantity $n$ with formula~(\ref{smallprob})
is through the period of the oscillations of $n$. Of course,
together with these large deviations, small fluctuations will be
present all of the time. A way of distinguishing both of them is
by means of their amplitude. The amplitude of small fluctuations
is $O\left(\sqrt{\mu/\sigma} \right)$ while the amplitude of these
oscillations promoted by large deviations is $O\left( \mu/\sigma
\right)$. So the difference should certainly be measurable in the
limit $\mu \gg \sigma$, which is exactly the range of validity of
our WKB approximation.

\subsection{Chemical chaos}

Not only periodicity but also chaotic trajectories are possible in
this simple system. To obtain them we allow the branching rate to
be a periodic function of time $\mu \to \mu(t)>0$. Note that, due
to the structure of system Eqs. (\ref{hamilton}), we could let
either $\mu$, $\sigma$ or both be time dependent and still reduce
the system to a $\mu$ time dependent one (while $\sigma$ remains
constant) by means of a change of the temporal variable. In this
case we deal with the system
\begin{eqnarray}
\label{hamilton11}
\dot{p} &=& \mu(t)(p-p^2)+\sigma(p^2-p)q,
\\
\dot{q} &=&
\mu(t)(2p-1)q+\frac{\sigma}{2}(1-2p)q^2. \label{hamilton22}
\end{eqnarray}
Herein we still can identify three invariant lines: $\{p=0\}$,
$\{p=1\}$ and $\{q=0\}$. The mean-field dynamics appears on the
$\{p=1\}$ line, and is expressed by the equation
\begin{equation}
\dot{q} = \mu(t) q - \frac{\sigma}{2} q^2.
\end{equation}
This differential equation is of Bernoulli type and can be solved
to yield
\begin{equation}
q(t)=\frac{ q(0) \exp \left[ \int_0^t \mu(t_1)dt_1
\right]}{1+\frac{\sigma}{2}q(0)\int_0^t \exp \left[ \int_0^{t_1}
\mu(t_2) dt_2 \right]dt_1}.
\end{equation}
Assuming that $\mu(t)=\mu+\epsilon h(t)$, where $h(t)$ is
$T-$periodic and continuous and $\epsilon$ is a small enough
constant, we know that there exists one $T-$periodic solution
$q_s(t)$ which attracts all initial conditions $q(0)>0$. This is
the solution whose initial condition fulfills
\begin{equation}
q(0)=\frac{\exp \left[ \int_0^T \mu(t_1)dt_1
\right]-1}{\frac{\sigma}{2}\int_0^T \exp \left[ \int_0^{t_1}
\mu(t_2) dt_2 \right]dt_1}.
\end{equation}
The points $(0,0)$ and $(1,0)$ continue to be fixed points in the
non-autonomous system, and the point $(0,2\mu/\sigma)$ gives rise
to a periodic orbit on $\{p=0\}$, which is formally identical to
$q_s(t)$, but it is unstable on this line, and we will refer to it
as $q_u(t)$. While the invariant line $\{q=2 \mu/\sigma\}$ is not
present in the perturbed system, the periodic trajectories
$q_s(t)$ and $q_u(t)$ are still connected. They correspond to
fixed points $\bar{q}_s$ and $\bar{q}_u$ of the Poincar\'{e} map
associated to the forced continuous dynamical system. One may
apply the theory developed in \cite{escudero2} to see that for
generic perturbations $h(t)$ the unstable manifold of $\bar{q}_s$
intersects the stable manifold of $\bar{q}_u$ and thus guarantees
the existence of a heteroclinic connection linking both fixed
points.

The irregular behavior of the system dynamics comes from the fact
that the periodic trajectories in the autonomous system may become
quasiperiodic or even chaotic in the perturbed one. This can be
justified by classical arguments from the Kolmogorov-Arnold-Moser
(KAM) theory \cite{A,SM} as follows.

Let us consider the stable equilibrium $P=(1/2,\mu/\sigma)$. The
Floquet multipliers are by definition the eigenvalues of the
corresponding Poincar\'e matrix. By the Hamiltonian structure, the
Floquet multipliers are complex conjugate numbers
$\lambda_1,\lambda_2$ such that $\lambda_1\lambda_2=1$. For $P$, a
direct computation on the linea\-rized system gives
$\lambda_1=e^{\omega iT}$,
$\lambda_2=\overline\lambda_1$, with $\displaystyle{\omega=\mu/2}$. The equilibrium $P$ is said to be
{\em non-degenerate} if $\displaystyle{\omega T\ne k\pi}$, for
$k=1,\ldots,4$. A non-degenerate equilibrium is persistent under
small perturbations as a fixed point of the Poincar\'e map. In
other words, $P$ is continued as a $T$-periodic solution of the
perturbed system (\ref{hamilton11})-(\ref{hamilton22}) for small
values of $\epsilon$. Besides, the presence of the heteroclinic
loop corresponding to the energy level $H=0$ in the unperturbed
system guarantees that the center around $P$ is not isochronous,
that is, the Poincar\'e map is of twist type. Under such
conditions, a generic perturbation gives rise to a classical KAM
scenario, composed by a dense set of invariant tori (corresponding
to quasiperiodic solutions) that are gradually destroyed as the
perturbation increases, giving rise to  domains of chaotic motion
(Smale horseshoes) intermingled with stability islands.

 Figure \ref{caos} shows a chaotic orbit surrounding a set of five stability
islands. Such chaotic orbits arises from the destruction of an invariant torus that
 persists until a critical value of the perturbation parameter $\epsilon$. Figure \ref{caos2}
presents a zoom of the latter picture, where the typical fractal
structure can be appreciated.

\begin{figure}[h]
\begin{center}
\includegraphics[scale=0.6]{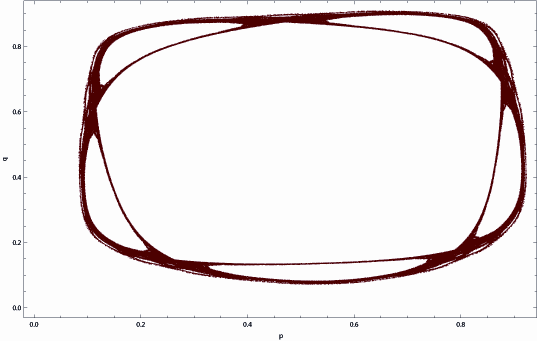}
\caption{Numerical simulation showing chaos for system
\eqref{hamilton11}-\eqref{hamilton22} with
$h(t)=\sin{(2t)},\mu=1,\sigma=2,\epsilon=0.97$. It is drawn the
Poincar\'e section of a single orbit with initial
 values $p(0)=0.53, q(0)=0.91$. More than $10^5$ points have been computed.}
\label{caos}
\end{center}
\end{figure}

\begin{figure}[h]
\begin{center}
\includegraphics[scale=0.6]{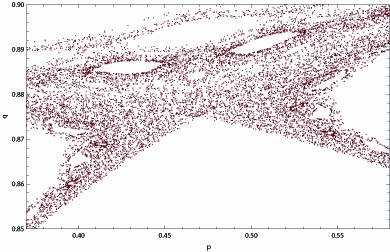}
\caption{Zoom of Fig. \ref{caos}. Besides the big stability
islands, some smaller holes can be appreciated, conforming to a
fractal structure.} \label{caos2}
\end{center}
\end{figure}

Let us mention that the stability islands are centered in periodic orbits of higher periods
(or subharmonic solutions). In the case of Figure \ref{caos},
a subharmonic solution of order 5 is located at the stability islands. Section 5 is devoted to the study
of the existence of such subharmonic solutions.

\begin{figure}
\centering \subfigure[]{
\includegraphics[width=0.4\textwidth]{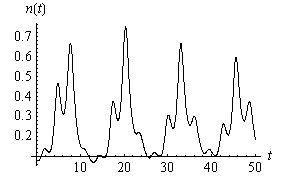}
\label{chaosn11}} \subfigure[]{
\includegraphics[width=0.4\textwidth]{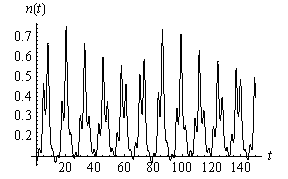}
\label{chaosn1}} \subfigure[]{
\includegraphics[width=0.4\textwidth]{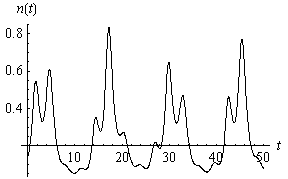}
\label{chaosn22}} \subfigure[]{
\includegraphics[width=0.4\textwidth]{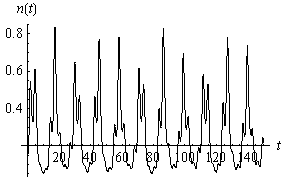}
\label{chaosn2} } \caption{Number of reagents $n(t)=p(t)q(t)$
versus time $t$ obtained from numerically integrating dynamical
system~(\ref{hamilton11})-(\ref{hamilton22}). The values of the
parameters are $\mu=1$, $\sigma=2$ and $\epsilon=0.97$; the
function $h(t)=\sin{(2t)}$. Panels \subref{chaosn11} and
\subref{chaosn1}: The system is initialized with the conditions
$p(0)=1/5$ and $q(0)=2/5$ which correspond to $n(0)=2/25$. Panels
\subref{chaosn22} and \subref{chaosn2}: The system is initialized
with the conditions $p(0)=3/5$ and $q(0)=1/4$ which correspond to
$n(0)=3/20$.} \label{chaoticn}
\end{figure}

We finish this section saying that we can compute the probability
with which a chaotic orbit appears by means of the exponentiated
negative of the action. This has already been done in explicitly
time dependent chemical systems for the simpler extinction
trajectories~\cite{meerson}. Herein we have shown that for
exponentially long times the sort of Hamiltonian chaos we have
described is possible for the simple reaction~(\ref{reactionset})
by virtue of chemical noise. In absence of noise only periodic
trajectories are possible. We have represented in
figure~\ref{chaoticn} the resulting aperiodic trajectories for the
number of reagents $n(t)=p(t)q(t)$ for two different time slots
and initial conditions. Their physical interpretation is analogous
to that of the periodic case in the last section. It is
interesting to note that some of the irregular motions that can be
observed in certain stochastic reaction dynamics might have an
underlying deterministic structure (we are always referring to
large deviations); let us recall that the time evolutions
represented in figure~\ref{chaoticn} are purely deterministic.

\section{Global continuation of the equilibrium point}
\vspace{0.5 cm} \noindent In subsection 3.2 we point out that the
equilibrium point $P=(1/2,\mu/\sigma)$ of system (\ref{hamilton})
is not degenerate if $\omega T \neq k\pi$, for $k=1,\dots, 4$,
with $\omega=\mu/2$. This property implies the continuation of $P$
as a $T$-periodic solution of the perturbed system
(\ref{hamilton11})-(\ref{hamilton22}) for small values of
$\epsilon$. In this section we find an explicit interval
$[0,E^{*}]$ for the parameter $\epsilon$ where this continuation
exists and is unique. The main result is inspired by \cite{Zh},
where a similar technique is applied to the classical pendulum
equation.

\vspace{0.5 cm}
\noindent
For simplicity in the calculations we apply the following change of variables to the perturbed system (\ref{hamilton11})-(\ref{hamilton22})
\begin{equation}\label{cambio de coordenadas}
T:\R^{2}\to \R^{2},  \quad T(p,q)=\Big(p-\frac{1}{2},q-\frac{\mu}{\sigma}\Big),
\end{equation}

\noindent
and obtain the new system
\begin{equation}\label{PS}
\left\{
\begin{aligned}
\dot{p}&=-\frac{\sigma}{4}q+\sigma p^{2}q+\epsilon \,h(t)(\frac{1}{4}-p^{2})\\
\dot{q}&=\frac{\mu^{2}}{\sigma}p-\sigma q^{2}p + 2\epsilon h(t)(q+\frac{\mu}{\sigma})p
\end{aligned}
\right.
\end{equation}

\noindent Note that for the perturbed system (\ref{PS}), the
invariant lines are
\[
\left\{p=-1/2\right\}, \quad \left\{p=1/2\right\}, \quad
\left\{q=-\mu/\sigma\right\}.
\]

\vspace{0.5 cm}
\noindent
In the following, we assume that $\displaystyle{\omega T/2 \notin \N}$ with $\omega=\mu/2$.

\begin{theorem}. There exists $\beta=\beta(\omega,\sigma)$ such that for $\epsilon \in [0,E^{*}[$ with
$\displaystyle{E^{*}=\frac{\sigma}{8\omega\beta h^{*}}}$, with
$h^{*}=\max_{t\in [0,T]}\left\{h(t)\right\}$, system (\ref{PS})
has a unique non-trivial $T$-periodic solution
$\Phi(t,\epsilon)=(\phi(t,\epsilon),\psi(t,\epsilon))$ as a
continuation of the equilibrium point $P=(0,0)$ of the autonomous
case.
\end{theorem}

\noindent \textbf{Remark.} From the proof,
$\beta=\beta(\omega,\sigma)$ is explicitly given by
\[
\beta=\beta(\omega,\sigma)=\displaystyle{\max_{t\in [0,T]}}\int_{0}^{T}|G(t,s)|ds,
\]
where $G$ is a Green's matrix function associated to the system
(\ref{PS}) defined below, $|.|$ means the usual uniform matrix
norm. The explicit bound
\[
\beta(\omega,\sigma)< \frac{T}{|\sin (\omega T/2)|} \max \left\{\frac{1}{2},\frac{2}{\sigma},\frac{\sigma}{8\omega}\right\}.
\]
is easily derived.
\noindent
\begin{proof}

\vspace{0.5 cm}
\noindent
We rewrite the system (\ref{PS}) in the form
\begin{equation}\label{PSvector}
\dot{X}=AX+B(t,X,\epsilon),
\end{equation}
with
\[
A=\begin{pmatrix}
0&-\displaystyle{\sigma/4}\\
\displaystyle{\mu^{2}/\sigma}&0
\end{pmatrix},\quad B(t,X,\epsilon)=\begin{pmatrix}
B_{1}(t,X,\epsilon)\\
B_{2}(t,X,\epsilon)
\end{pmatrix},
\]
and
\begin{equation*}
\begin{split}
B_{1}(t,X,\epsilon)&=\displaystyle{\sigma p^{2}q+\epsilon\, h(t)(\frac{1}{4}-p^{2})},\\
B_{2}(t,X,\epsilon)&=-\sigma q^{2}p+2\epsilon\, h(t)(q+\frac{\mu}{\sigma})p.
\end{split}
\end{equation*}
Now if $X(t,\epsilon)$ is a $T$-periodic solution of (\ref{PSvector}) the method of variation of constants provied us the next formula
\begin{equation}\label{formula}
X(t,\epsilon)=\int_{0}^{T}G(t,s)B(s,X(s,\epsilon),\epsilon) ds,
\end{equation}
where $G(t,s)$ is the Green's matrix function associated to this problem given by
\[
G(t,s)=\begin{cases}
\displaystyle{J^{-1}e^{(t-s)A}}; & \textit{if} \quad 0\leq s\leq t\leq T\\
\displaystyle{J^{-1}e^{T A}e^{(t-s)A}}; & \textit{if} \quad 0\leq t\leq s\leq T
\end{cases}
\]
\noindent
where the matrix $J$ is defined by
\[
J=(I_{2}-e^{T A})=\begin{pmatrix}
\displaystyle{1-\cos \omega T}&\displaystyle{\frac{\sigma \omega}{4}\sin \omega T} \\
\displaystyle{-\frac{4 \omega}{\sigma} \sin \omega T}& \displaystyle{1-\cos \omega T}
\end{pmatrix}.
\]

Explicitly, the Green's matrix is
\[
G(t,s)=\frac{1}{2}\begin{pmatrix}
\displaystyle{\frac{\sin \omega(T/2-(t-s))}{\sin(\omega T/2)}}&-\displaystyle{\frac{\sigma}{4\omega}\frac{\cos \omega(T/2-(t-s))}{\sin(\omega T/2)}} \\
\displaystyle{\frac{4\omega}{\sigma}\frac{\cos \omega(T/2-(t-s))}{\sin(\omega T/2)}}& \displaystyle{\frac{\sin \omega(T/2-(t-s))}{\sin(\omega T/2)}}
\end{pmatrix},
\]
\noindent for all $0\leq s\leq t\leq T$  and
\[
G(t,s)=\frac{1}{2}\begin{pmatrix}
\displaystyle{\frac{\sin \omega((s-t)-T/2)}{\sin(\omega T/2)}}&-\displaystyle{\frac{\sigma}{4\omega}\frac{\cos \omega((s-t)-T/2)}{\sin(\omega T/2)}} \\
\displaystyle{\frac{4\omega}{\sigma}\frac{\cos \omega((s-t)-T/2)}{\sin(\omega T/2)}}& \displaystyle{\frac{\sin \omega((s-t)-T/2)}{\sin(\omega T/2)}}
\end{pmatrix}
\]
for $0\leq t\leq s\leq T$. In Appendix 1 we present the explicit
calculation of (\ref{formula}).

\vspace{0.5 cm}
\noindent
Consider $\Omega=\R\times [0,E^{*}]$ and the normed space
\begin{equation*}
E=\left\{X\in C(\Omega,\R^{2}):\textit{$X$ is $T$-periodic} \right\},
\end{equation*}
with the norm $\left\|\cdot \right\|_{\infty}$. We define the operator $\mathcal{T}:E\to E$ given by
\[
(\mathcal{T}X)(t):=\int_{0}^{T}G(t,s)B(s,X(s,\epsilon),\epsilon)\, ds,
\]
which is a completely continuous operator (with the norm $\left\|\cdot \right\|_{\infty}$ ) from $E$ to itself. It follows from (\ref{formula}) that $X(t,\epsilon)$ is a $T$-periodic solution of (\ref{PSvector}) if and only if $X$ is a fixed point of $\mathcal{T}$.

\noindent
Now we concentrate on estimating a value $E^{*}$ where the operator $\mathcal{T}$ will be a contraction a let invariant a closed ball $\mathcal{B}=\left\{X\in E: \left\|X \right\|_{\infty}\leq \rho\right\}$ for some positive number $\rho=\rho(\omega,\sigma, \epsilon)$. To this end, let $X=(p,q)$, $Y=(\widetilde{p},\widetilde{q})$ inside the ball $\mathcal{B}$. Observe that
\begin{equation*}
\begin{split}
\left\|(\mathcal{T}X)(t)-(\mathcal{T}Y)(t)\right\|_{\infty}&=\left\|\int_{0}^{T}G(t,s)(B(s,X,\epsilon)-B(s,Y,\epsilon)ds\right\|_{\infty}\\
&\leq \Big(\displaystyle{\max_{t\in [0,T]}}\int_{0}^{T}|G(t,s)|ds \Big) \left\|B(s,X,\epsilon)-B(s,Y,\epsilon) \right\|_{\infty}
\end{split}
\end{equation*}
where
\[
\beta=\beta(\mu,\sigma)=\displaystyle{\max_{t\in [0,T]}}\int_{0}^{T}|G(t,s)|ds =\max_{1\leq i,j\leq 2}\Big(\max_{t\in [0,T]}\int_{0}^{T}|G_{i,j}(t,s)|ds\Big).
\]
Next we consider
\[
B_{1}(t,X,\epsilon)-B_{1}(t,Y,\epsilon)=\displaystyle{\epsilon h(t)(\widetilde{p}^{2}-p^{2})+\sigma (p^{2}q-\widetilde{p}^{2}\widetilde{q})}.
\]
Notes that
\[
|\epsilon h(t)(\widetilde{p}^{2}-p^{2})|\leq 2\epsilon h^{*}\rho \left\|X-Y\right\|_{\infty}, \quad \text{with} \quad h^{*}=\max_{t\in [0,T]}\left\{h(t)\right\},
\]
and
\begin{equation*}
|p^{2}q-\widetilde{p}^{2}\widetilde{q}|=|(p-\widetilde{p})(pq+\widetilde{p}q)+\widetilde{p}^{2}(q-\widetilde{q})|\leq 3\rho^{2} \left\|X-Y\right\|_{\infty}
\end{equation*}
therefore
\[
\left\|B_{1}(t,X,\epsilon)-B_{1}(t,Y,\epsilon)\right\|_{\infty}\leq 2\rho(\epsilon h^{*}+ 3\sigma \rho/2)\left\|X-Y\right\|_{\infty}.
\]

\noindent
On the other hand
\[
|B_{2}(t,X,\epsilon)-B_{2}(t,Y,\epsilon)|=|2\epsilon h(t)(qp-\widetilde{q}\widetilde{p})+\sigma(\widetilde{q}^{2}\widetilde{p}-q^{2}p)+\frac{2\epsilon h(t)\mu}{\sigma}(p-\widetilde{p})|.
\]
For the first two terms in the right hand we have that
\[
|\widetilde{q}^{2}\widetilde{p}-q^{2}p|=|(\widetilde{q}-q)(\widetilde{q}+p)p+\widetilde{q}^{2}\widetilde{p}-p\widetilde{q}|\leq 3\rho^{2}\left\|X-Y\right\|_{\infty},
\]
and
\[
|qp-\widetilde{q}\widetilde{p}|=|p(q-\widetilde{q})+(p-\widetilde{p})\widetilde{q}|\leq 2\rho\left\|X-Y\right\|_{\infty}
\]
In consequence
\[
\left\|B_{2}(t,X,\epsilon)-B_{2}(t,Y,\epsilon)\right\|_{\infty}\leq \Big(2\rho(\epsilon h^{*}+3\sigma\rho/2)+2\epsilon h^{*}(\rho+\mu/\sigma)\Big)\left\|X-Y\right\|_{\infty}
\]

\noindent This estimate implies that
\[
\left\|B(t,X,\epsilon)-B(t,Y,\epsilon)\right\|_{\infty}\leq \Big(2\rho(\epsilon h^{*}+3\sigma\rho/2)+2\epsilon h^{*}(\rho+\mu/\sigma)\Big)\left\|X-Y\right\|_{\infty}
\]
Finally
\[
\left\|\mathcal{T}X-\mathcal{T}Y\right\|_{\infty}\leq \beta \Big(2\rho(\epsilon h^{*}+3\sigma\rho/2)+2\epsilon h^{*}(\rho+\mu/\sigma)\Big)\left\|X-Y\right\|_{\infty}
\]

\noindent
From this last inequality, we concluded that the operator $\mathcal{T}$ will be a contraction if
\[
\frac{3}{2}\sigma\rho^{2} +4\epsilon h^{*}\rho +\frac{2}{\sigma}\epsilon h^{*}\mu < \frac{1}{\beta}.
\]
Now we use the fact that the quadratic convex function
\[
f(\rho)=\frac{3}{2}\sigma\rho^{2} +4\epsilon h^{*}\rho +\frac{2}{\sigma}\epsilon h^{*}\mu -\frac{1}{\beta},
\]
is negative in $[0,\rho_{0}[$ if $f(0)<0$, (this is equivalent to have $\displaystyle{\epsilon < \frac{\sigma}{4\mu\beta h^{*}}}$)  with
\[
\rho_{0}=\frac{-4\epsilon h^{*}+2\sqrt{\displaystyle{(2\epsilon h^{*})^{2}-3\sigma\Big(\frac{\epsilon \mu h^{*}}{\sigma}-\frac{1}{2\beta}\Big)}}}{3\sigma}.
\]

\noindent
\vspace{0.5 cm}
In this way for $\epsilon \in [0,E^{*}[$ with $\displaystyle{E^{*}=\frac{\sigma}{4\mu\beta h^{*}}}$  we have
\[
\left\|\mathcal{T}X-\mathcal{T}Y\right\|_{\infty}\leq k\left\|X-Y\right\|_{\infty}
\]
with $0<k<1$ for all $X,Y \in \mathcal{B}$. Remains to knows the size of radius of the ball $\mathcal{B}$. Let $Y_{0}=(0,0)$, then
\[
(\mathcal{T}Y_{0})(t)=\frac{\epsilon}{4}\begin{pmatrix}
\displaystyle{\int_{0}^{T}G_{1,1}(t,s)h(s)ds}\\
\displaystyle{\int_{0}^{T}G_{2,1}(t,s)h(s)ds}
\end{pmatrix},
\]
moreover
\[
\left\|\mathcal{T}X\right\|_{\infty}-\left\|\mathcal{T}Y_{0}\right\|_{\infty}\leq \left\|\mathcal{T}X-\mathcal{T}Y_{0}\right\|_{\infty}\leq k \left\|X-Y_{0}\right\|_{\infty}
\]
this implies
\[
\left\|\mathcal{T}X\right\|_{\infty}\leq k\left\|X\right\|_{\infty}+\left\|\mathcal{T}Y_{0}\right\|_{\infty}=k\rho+\epsilon h^{*}\beta/4.
\]

\noindent
We take $\displaystyle{\rho=\frac{\epsilon h^{*}\beta}{4(1-k)}}$ and $\epsilon$ small enough sucht that $f(\rho)<0$. Combining these estimatives we have that
\[
\left\|\mathcal{T}X\right\|_{\infty}\leq \rho, \quad \text{and} \quad \left\|\mathcal{T}X-\mathcal{T}Y\right\|_{\infty}<k\left\|X-Y\right\|_{\infty},
\]
for all $X,Y \in \mathcal{B}.$ Then $\mathcal{T}$ maps the closed
ball $\mathcal{B}$ into itself. Thus it follows from the Schauder
fixed point theorem \cite{Schauder} that $\mathcal{T}$ has a fixed
point $\Phi$ in $\mathcal{B}$. Since $\mathcal{T}: \mathcal{B}\to
\mathcal{B}$ is a contraction then $\mathcal{T}$ the fixed point
is unique.
\end{proof}

\section{Local continuation of periodic solutions from the autonomous case}

\noindent From Section 3 we know that $P=(1/2,\mu/\sigma)$ is an
equilibrium point for the autonomous system (\ref{hamilton}).
Moreover $P$ is a center and therefore there is a domain of
periodic solutions around $P$.  In Section 4, after a change of
coordinates (\ref{cambio de coordenadas}), the point $P$ is
extended for $\epsilon \in \, ]0,E^{*}]$ as a $T$-periodic
solution $\Phi(t,\epsilon)$ of (\ref{PS}).  Furthermore, this
extension is unique and continuous.






\vspace{0.5 cm} \noindent In this section, we look for
multiplicity of periodic solutions. To this purpose, we assume
that $h(t)$ is an even function. Under this assumption, (\ref{PS})
has the following symmetry
\[
(t,p,q)\to (-t,-p,q).
\]
Our aim is to obtain $nT$-periodic solutions of (\ref{PS}) as a
result of the local continuation of $nT$-periodic solutions of the
autonomous Hamiltonian system (\ref{hamilton}). To this end we
present the next result which is inspired by the results on
\cite[Section 5]{Llibre-Ortega}.

Some notation is needed. $[\cdot]$ denotes the integer part
function. Fix $n_{*}=\displaystyle{\big[\frac{2\pi}{\omega\,
T}]}$. For a fixed integer $n\geq n_{*}$, define
$\displaystyle{\vartheta_{n}=\big[\frac{\omega\, nT}{2\pi}]}$.

\begin{theorem} For all $n>n_{*}$ and $m=1,\dots, \vartheta_{n}$, there exists $\epsilon_{n,m}>0$ such that for all $0<\epsilon <
\epsilon_{n,m}$ system (\ref{PS}) has a non-trivial $nT$-periodic
solution $X(t,\epsilon)=(p(t,\epsilon),q(t,\epsilon))$ where
$q(t,\epsilon)$ crosses exactly $m$ times through the horizontal
line $q=0$  in the interval  $[0,nT/2]$.
\end{theorem}

\begin{proof}  Let $X(t;\xi,0)=(p(t;\xi,0),q(t,\xi,0))$ a solution (\ref{PS}) in the autonomous case ($\epsilon=0$)  satisfying the initial condition
\begin{equation}\label{icps}
p(0,\xi,0)=0, \quad q(0,\xi,0)=\xi,
\end{equation} and the boundary condition
\begin{equation}\label{ibcps}
p(nT/2,\xi,0)=0.
\end{equation}

In Appendix 2 we prove that the family of periodic solutions
$X(t;\xi,0)$ of (\ref{PS}) for $0<\xi< \mu/ \sigma$ has an
increasing period function $T(\xi)$  such that
\[
\lim_{\xi \searrow 0} T(\xi)=T_{lp} \quad \text{and} \quad  \lim_{\xi \nearrow \mu/ \sigma} T(\xi)=\infty,
\]
\noindent
with $\displaystyle{T_{lp}=2\pi/\omega}$. Moreover, in the autonomous case we have the following symmetry
\[
(t,p,q)\to (-t,p,-q),
\]
therefore we can assume that $\xi>0$. Given $n\in \N$, $X(t;\xi,0)$ is a $nT$-periodic solution of (\ref{PS}) if and only if there is an integer $m \geq 1$ such that
\begin{equation}\label{vp}
m\,T(\xi)=nT.
\end{equation}
Since
\[
\inf T(\xi)=T_{lp}=\frac{2\pi}{\omega},
\]
we  have  $\displaystyle{\frac{nT}{m}>\frac{2\pi}{\omega}}$ therefore $\displaystyle{m<\frac{\omega\, nT}{2\pi}}$.

\vspace{0.5 cm}
\noindent
Let
\[
\mu/\sigma>\xi_{1}>\xi_{2}> \dots >\xi_{\vartheta_{n}}>0,
\]
be the solutions of (\ref{vp}) with $m=1,2,\dots, \vartheta_{n}$. Since we consider the boun\-dary condition (\ref{ibcps})
Then
\begin{equation*}
\begin{split}
\mathcal{K}_{0}&=\left\{\xi \in \R: p(n T/2,\xi,0)=0\right\}\\
&=\{-\xi_{1},\dots,-\xi_{\vartheta_{N}},\,0\,,\xi_{1},\dots, \xi_{\vartheta_{N}}\}.\\
\end{split}
\end{equation*}

\noindent
Now we compute the index for $\xi_{1} \in
\mathcal{K}_{0} $. Note the following
\begin{itemize}
    \item If $\xi<\xi_{1}$ then $\displaystyle{\frac{T(\xi)}{2}<\frac{T(\xi_{1})}{2}=n T/2},$
    \item If $\xi>\xi_{1}$ then $\displaystyle{\frac{T(\xi)}{2}>\frac{T(\xi_{1})}{2}=n T/2}.$
\end{itemize}

Since $T(\xi)$ is increasing, we have for values $\xi$ close to $\xi_{1}$ that
\begin{equation*}
\begin{split}
    p(nT/2;\xi,0)>0 & \quad \text{if} \quad \xi<\xi_{1},\\
    p(nT/2;\xi,0)<0 & \quad \text{if} \quad \xi>\xi_{1}.
\end{split}
\end{equation*}

\noindent From here, the Brouwer index
$\text{ind}(p(nT/2,\cdot,0),\xi_{1})=-1$ (see for instance
\cite{De} for definition and basic properties). From the previous
calculus we can conclude that in general
\begin{equation}\label{ecuaci¾n-periodo}
\text{ind}((p(nT/2,\cdot,0),\xi_{m})=(-1)^{m}.
\end{equation}
By symmetry the indices of $-\xi_{m}$ are
$\text{ind}(p(nT/2,\cdot,0),\xi_{m})=(-1)^{m+1}$. We also
calculate the index at $\xi_{0}=0$. We do this by linearization,
i.e.,
\[
\text{ind}((p(nT/2,\cdot,0),0)=\text{sign}\Big(\frac{\partial p}{\partial \xi}(nT/2;0,0)\Big).
\]

To this end we consider the linearized problem of at $(0,\xi)$ and we observe that $\displaystyle{\frac{\partial X}{\partial \xi}(t,\xi,0)}$ is solution of the initial value problem
\[
\dot{X}=AX, \quad X(0)=\begin{pmatrix}
0\\
1
\end{pmatrix},
\]
therefore
\[
\text{ind}((p(nT/2,\cdot,0),0)=\text{sign}\Big(\frac{4\omega}{\sigma}\sin\big(\frac{\omega\, n T}{2}\big)\Big)=(-1)^{\vartheta_{n}}.
\]

\noindent
In conclusion for $\epsilon=0$ there exists $\vartheta_{n}$ nontrivial
$nT$-periodic solutions of (\ref{nps}) with
\[
p(0;\xi,0)=0, \quad q(0;\xi,0)=\xi>0,
\]

\noindent
such solutions can be labeled according to the number of times $m$ that the function $\displaystyle{q(t;\xi,0)}$ passes through the horizontal line $q=0$  in $[0,nT/2]$ with $m=1,\ldots,\vartheta_{n}$ and initial conditions
\[
\mu/\sigma>q_{1}(0)=\xi_{1}>\cdots
>q_{\vartheta_{n}}(0)=\xi_{\vartheta_{n}}>0.
\]

\noindent
From (\ref{ecuaci¾n-periodo}) and the Implicit Function theorem, for each $n\geq n_{*}$ with $\displaystyle{n_{*}=[\frac{2\pi}{\omega T}]}$ the is a $C^{1}$ function $E:[0,\epsilon_{n,m}[\to \R$, $\epsilon \to E(\epsilon)$ such that $X(t,\epsilon)$ is a $nT$-periodic solution of (\ref{PS}) that satisfy the initial condition
\[
X(0,\epsilon)=\begin{pmatrix}
0\\
E(\epsilon)
\end{pmatrix},
\]
and the boundary condition
\[
p(nT/2,E(\epsilon),\epsilon)=0,
\]
for all $\epsilon \in [0,\epsilon_{n,m}[.$ This completes the proof.

\end{proof}

\section{Conclusions and outlook}

In this work we have shown the appearance of chaotic and periodic
behavior in chemical systems as a direct consequence of internal
fluctuations. We have concentrated on the simple reaction $A
\longleftrightarrow 2A$, which at the mean-field (fluctuations
free) level simply shows logistic growth. While this theory
correctly predicts the short time behavior, long times are
dominated by fluctuations. In this case we have seen that
fluctuations may sustain metastable states and periodic orbits.
For this same reaction, if we allow the reaction rates to vary
periodically in time we find the presence of chaotic orbits
sustained by chemical noise, while the mean-field theory only
reflects periodic orbits. We have also been able to rigorously
prove the existence of even and periodic solutions of the
nonautonomous system as a result of the global continuation of
even and periodic solutions of the autonomous system. It is
important to remark here that the deterministic trajectories
studied here will not be observed in experiments, but their noisy
counterparts. That is, small Gaussian distributed fluctuations
about the deterministic trajectories studied herein will be
present at any time. We also note that the appearance of
periodicity and chaos in chemical reactions due to intrinsic
fluctuations has already been
investigated~\cite{kapral,kapral2,ross}. Nevertheless our approach
is fundamentally different as these previous studies focused on
systems close to a bifurcation threshold. So the periodic and
chaotic behaviors were already present in the mean-field
deterministic dynamics, and the internal noise anticipated the
threshold. In our case, the mean-field dynamics were of too low
dimensionality for showing periodic and chaotic orbits
respectively. In this sense, the oscillations and chaos were
purely sustained by chemical fluctuations.

\begin{figure}[h]
\begin{center}
\includegraphics[scale=0.7]{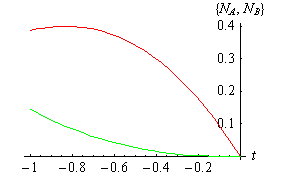}
\caption{Optimal paths to extinction in Lotka-Volterra
dynamics~(\ref{lotkavolterraeq}) for $\mu=\sigma=\lambda=1$. The
initial numbers of predators (represented by the red upper line)
and preys (represented by the green lower line) are respectively
$N_A(t=-1) \approx 0.39$ and $N_B(t=-1) \approx 0.15$. The
extinction takes place at $t=0$ and for $p_a=p_b=0$, $q_a=-1$ and
$q_b=1$. The trajectories lie in $H=1$ manifold.}
\label{lotkavolterra}
\end{center}
\end{figure}

We have also studied the optimal paths to extinction in a plankton
population dynamics model. Although extinction phenomena have been
considered numerous times within this
framework~\cite{assaf1,assaf2,meerson}, the approach was based on
the Hamiltonian dynamics on the stationary $H=0$ manifold. The
Hamiltonian dynamical system in our case was however degenerated
and all the extinction trajectories fell on the $H \neq 0$
manifolds. We leave for future work the extension of our results
to multispecies reactions, which are characterized by large
deviation Hamiltonian dynamical systems of higher dimension.
Notably, there is an important problem of a two species reaction
which is very much related to the plankton extinction described
herein. It is the appearance of extinction events in the
predator-prey Lotka-Volterra dynamics. This dynamics is formalized
by means of the reaction set
\begin{equation}
A \stackrel{\mu}\longrightarrow \emptyset, \quad B
\stackrel{\sigma}\longrightarrow 2B, \quad A+B
\stackrel{\lambda}\longrightarrow 2A,
\end{equation}
where $A$ is the predator and $B$ is the prey. This system has
been studied by means of a diffusion approximation and solving the
corresponding Fokker-Planck equation~\cite{parker}. A different
alternative is using a large deviations approach, which yields the
Hamiltonian
\begin{equation}
\label{lotkavolterraeq}
\mathcal{H}= \mu (p_a-1) q_a-\sigma p_b
(p_b-1) q_b +\lambda p_a(p_b-p_a)q_a q_b.
\end{equation}
The numbers of predator and prey are given respectively by $N_A(t)
= p_a(t) q_a(t)$ and $N_B(t) = p_b(t) q_b(t)$. The optimal paths
to extinction appear again on $H \neq 0$ manifolds. We have
numerically computed two of them in Fig.~\ref{lotkavolterra}. A
large deviation drives the system to extinction which takes place
when $p_a$ and $p_b$ become zero. One can see in this figure that
the decreasing number of preys pulls the predators to extinction,
as it is reasonable to expect. Alternatively, a large deviation
driving the predators to extinction will lead to an unbounded
growth of the preys. We expect that a systematic study of the
extinction trajectories of the Lotka-Volterra
Hamiltonian~(\ref{lotkavolterraeq}) will shed a valuable light on
the probabilistic structure of extinction events in predator-prey
dynamics.

\section*{Acknowledgments}

The authors are grateful to Alex Kamenev and Baruch Meerson for
helpful discussions are correspondence respectively; they are also
grateful to one anonymous referee for his/her valuable comments.
A. Rivera and P.J. Torres are grateful to Rafael Ortega for useful
insights and for pointing out the reference \cite{Zh}. Carlos
Escudero is grateful to the Departamento de Mat\'ematica Aplicada
of the Universidad de Granada for its hospitality. This work has
been partially supported by the MICINN (Spain) through Project No.
MTM2008-02502.

\section*{Appendix 1}

The purpose of this appendix is to show more explicitly the calculations to obtain the expression of the operator $\mathcal{T}$ in the formula (\ref{formula}). To this end, consider the system
\begin{equation}\label{As}
\dot{X}=AX+B(t,X)
\end{equation}
with
\[
A=\begin{pmatrix}
0&-\displaystyle{\sigma/4}\\
\displaystyle{\mu^{2}/\sigma}&0
\end{pmatrix},\quad B(t,X)=\begin{pmatrix}
B_{1}(t,X)\\
B_{2}(t,X)
\end{pmatrix},
\]

\noindent
The fundamental matrix of the associated autonomous system $\dot{X}=AX$ is the exponential matrix $e^{tA}$ given by
\[
e^{tA}=\begin{pmatrix}
\cos \omega t &-\displaystyle{\frac{\sigma}{4\omega}}\sin \omega t\\
\displaystyle{\frac{4\omega}{\sigma}}\sin \omega t&\cos \omega t
\end{pmatrix}.
\]

\noindent
Then by the method of variations of constants the general solution of $(\ref{As})$ is given by
\begin{equation}\label{As-2}
\begin{split}
X(t)=e^{tA}C+\int_{0}^{t}e^{(t-s)A}B(s,X)\,ds,
\end{split}
\end{equation}

\noindent
where $C$ a constant vector. Since we are looking for $T$-periodic solutions imposing the boundary conditions $X(0)=X(T)$ we get
\[
C=e^{TA}C+\int_{0}^{T}e^{(t-s)A}B(s,X)\,ds
\]
this implies
\begin{equation*}
\begin{split}
(I_{2}-e^{TA})C&=\int_{0}^{T}e^{(T-s)A}B(s,X)\,ds\\
C&=(I_{2}-e^{TA})^{-1}\int_{0}^{T}e^{(T-s)A}B(s,X)\,ds,
\end{split}
\end{equation*}
replacing this in the formula $(\ref{As-2})$ we obtain
\begin{equation*}
X(t)=e^{tA}(I_{2}-e^{TA})^{-1}\int_{0}^{T}e^{(T-s)A}B(s,X)\,ds + \int_{0}^{t}e^{(t-s)A}B(s,X)\,ds,
\end{equation*}
On the other hand
\[
J=(I_{2}-e^{TA})=\begin{pmatrix}
1-\cos \omega T&\displaystyle{\frac{\sigma}{4\omega}}\sin \omega T\\
-\displaystyle{\frac{4\omega}{\sigma}}\sin \omega T&1-\cos \omega T
\end{pmatrix},
\]
since $\det{J}=2(1-\cos \omega T)$ then
\[
J^{-1}=\frac{1}{2}\begin{pmatrix}
1&-\displaystyle{\frac{\sigma}{4\omega}}\frac{\cos (\omega T/2)}{\sin (\omega T/2)}\\
\displaystyle{\frac{4\omega}{\sigma}}\frac{\cos (\omega T/2)}{\sin (\omega T/2)}&1
\end{pmatrix},
\]

\noindent
by direct calculation it is found that
\[
e^{tA}J^{-1}=J^{-1}e^{tA},\qquad e^{tA}e^{(T-s)A}=e^{TA}e^{(t-s)A}.
\]
Therefore
\begin{equation*}
\begin{split}
X(t)&=J^{-1}e^{TA}\int_{0}^{T}e^{(T-s)A}B(s,X)\,ds + \int_{0}^{t}e^{(t-s)A}B(s,X)\,ds\\
X(t)&=[J^{-1}e^{TA}+I_{2}]\int_{0}^{t}e^{(t-s)A}B(s,X)\,ds + J^{-1}e^{TA}\int_{t}^{T}e^{(t-s)A}B(s,X)\,ds
\end{split}
\end{equation*}

\noindent
The matrix $J$ satisfies $J^{-1}=J^{-1}e^{TA}+I_{2}$. In consequence
\[
X(t)=J^{-1}\int_{0}^{t}e^{(t-s)A}B(s,X)\,ds + J^{-1}e^{TA}\int_{t}^{T}e^{(t-s)A}B(s,X)\,ds
\]

\noindent
\vspace{0.5 cm}
We define the Green's matrix
\[
G(t,s)=\begin{cases}
\displaystyle{J^{-1}e^{(t-s)A}}; & \textit{if} \quad 0\leq s\leq t\leq T\\
\displaystyle{J^{-1}e^{T A}e^{(t-s)A}}; & \textit{if} \quad 0\leq t\leq s\leq T
\end{cases}
\]

Thus we can write
\[
X(t)=\int_{0}^{t}G(t,s)B(s,X)\,ds
\]

Let $G_{1}=J^{-1}e^{(t-s)A}$. Explicitly this matrix is given by
\begin{equation*}
G_{1}=\begin{pmatrix}
\frac{\cos \omega(t-s) \sin (\omega T/2)-\cos(\omega T/2)\sin \omega(t-s)}{2\sin (\omega T/2)}&-\frac{\sigma}{4\omega}[\frac{\sin \omega(t-s) \sin (\omega T/2)+\cos(\omega T/2)\cos \omega(t-s)}{2\sin (\omega T/2)}]\\
\frac{4\omega}{\sigma}[\frac{\cos \omega(t-s) \cos (\omega T/2)-\sin(\omega T/2)\sin \omega(t-s)}{2\sin (\omega T/2)}]&\frac{\cos \omega(t-s) \sin (\omega T/2)-\cos(\omega T/2)\sin \omega(t-s)}{2\sin (\omega T/2)}
\end{pmatrix}
\end{equation*}

From the trigonometric identities follows easily

\begin{equation*}
J^{-1}e^{(t-s)A}=\frac{1}{2}\begin{pmatrix}
\displaystyle{\frac{\sin \omega(T/2-(t-s))}{\sin(\omega T/2)}}&-\displaystyle{\frac{\sigma}{4\omega}\frac{\cos \omega(T/2-(t-s))}{\sin(\omega T/2)}} \\
\displaystyle{\frac{4\omega}{\sigma}\frac{\cos \omega(T/2-(t-s))}{\sin(\omega T/2)}}& \displaystyle{\frac{\sin \omega(T/2-(t-s))}{\sin(\omega T/2)}}
\end{pmatrix},
\end{equation*}

\vspace{0.5 cm}
Let $G_{2}=J^{-1}e^{T A}e^{(t-s)A}$. Explicitly this matrix is given by
\begin{equation*}
G_{2}=\begin{pmatrix}
-\frac{1}{2}&-\frac{\sigma}{4\omega}\frac{\cos(\omega T/2)}{2\sin (\omega T/2)}\\
\frac{4\omega}{\sigma}\frac{\cos(\omega T/2)}{2\sin (\omega T/2)}&-\frac{1}{2}
\end{pmatrix}\begin{pmatrix}
\cos \omega(t-s)&-\frac{\sigma}{4\omega}\sin \omega(t-s)\\
\frac{4\omega}{\sigma}\sin \omega(t-s)& \cos \omega(t-s)
\end{pmatrix}.
\end{equation*}
By direct calculation
\begin{equation*}
G_{2}=\begin{pmatrix}
\frac{-\cos \omega(t-s) \sin (\omega T/2)-\cos(\omega T/2)\sin \omega(t-s)}{2\sin (\omega T/2)}&-\frac{\sigma}{4\omega}[\frac{\cos(\omega T/2)\cos \omega(t-s)-\sin \omega(t-s)\sin (\omega T/2)}{2\sin (\omega T/2)}]\\
\frac{4\omega}{\sigma}[\frac{\cos(\omega T/2)\cos \omega(t-s)-\sin \omega(t-s)\sin (\omega T/2)}{2\sin (\omega T/2)}]&\frac{-\cos \omega(t-s) \sin (\omega T/2)-\cos(\omega T/2)\sin \omega(t-s)}{2\sin (\omega T/2)}
\end{pmatrix},
\end{equation*}
again, using basic trigonometric identities follows that
\begin{equation*}
J^{-1}e^{T A}e^{(t-s)A}=\frac{1}{2}\begin{pmatrix}
\displaystyle{\frac{\sin \omega((s-t)-T/2)}{\sin(\omega T/2)}}&-\displaystyle{\frac{\sigma}{4\omega}\frac{\cos \omega((s-t)-T/2)}{\sin(\omega T/2)}} \\
\displaystyle{\frac{4\omega}{\sigma}\frac{\cos \omega((s-t)-T/2)}{\sin(\omega T/2)}}& \displaystyle{\frac{\sin \omega((s-t)-T/2)}{\sin(\omega T/2)}}
\end{pmatrix}.
\end{equation*}

In orther to present a upper bound for the constant $\beta(\omega,\sigma)$ we give the calculations for
\[
\max_{t\in [0,T]}\int_{0}^{T}|G_{1,1}(t,s)|ds
\]
To this end, observe that
\[
\int_{0}^{T}|G_{1,1}(t,s)|ds =\frac{1}{2}\Big[\int_{0}^{t}\Big|\frac{\sin \omega(T/2-(t-s))}{\sin(\omega T/2)}\Big|\,ds + \int_{t}^{T}\Big|\frac{\sin \omega(T/2-(s-t))}{\sin(\omega T/2)}\Big|\,ds\Big].
\]

On the other hand
\[
\frac{1}{2}\int_{0}^{t}\Big|\frac{\sin \omega(T/2-(t-s))}{\sin(\omega T/2)}\Big|\,ds=\frac{1}{2\omega}\int_{\omega(T/2-t)}^{\omega T/2}\Big|\frac{\sin u}{\sin(\omega T/2)}\Big|\,du,
\]

and also
\[
\int_{t}^{T}\Big|\frac{\sin \omega(T/2-(s-t))}{\sin(\omega T/2)}\Big|\,ds=-\frac{1}{2\omega}\int_{\omega T/2}^{\omega(t-T/2)}\Big|\frac{\sin u}{\sin(\omega T/2)}\Big|\,du.
\]
It follows that
\[
\int_{0}^{T}|G_{1,1}(t,s)|ds =\frac{1}{\omega}\int_{0}^{\omega T/2}\frac{|\sin u|}{|\sin(\omega T/2)|}\,ds
\]
The same analysis applies to other components of the matrix $G$ and we obtain
\begin{equation*}
\begin{split}
\int_{0}^{T}|G_{1,2}(t,s)|ds &=\frac{\sigma}{4\omega^{2}}\int_{0}^{\omega T/2}\frac{|\cos u|}{|\sin(\omega T/2)|}\,ds,\\
\int_{0}^{T}|G_{2,1}(t,s)|ds &=\frac{4}{\sigma}\int_{0}^{\omega T/2}\frac{|\cos u|}{|\sin(\omega T/2)|}\,ds,\\
\int_{0}^{T}|G_{2,2}(t,s)|ds &=\frac{1}{\omega}\int_{0}^{\omega T/2}\frac{|\sin u|}{|\sin(\omega T/2)|}\,ds.
\end{split}
\end{equation*}

Since
\[
\beta=\beta(\mu,\sigma)=\displaystyle{\max_{t\in [0,T]}}\int_{0}^{T}|G(t,s)|ds =\max_{1\leq i,j\leq 2}\Big(\max_{t\in [0,T]}\int_{0}^{T}|G_{i,j}(t,s)|ds\Big),
\]
follows that
\[
\beta(\omega,\sigma)< \frac{T}{|\sin \omega T/2|} \max \left\{\frac{1}{2},\frac{2}{\sigma},\frac{\sigma}{8\omega}\right\}.
\]

\section*{Appendix 2}

In this appendix we are going to study the period function of periodic solutions around
the center $(1/2,\mu/\sigma)$ for the autonomous case. Given the clear symmetry of the
vector field of our Hamiltonian system
(\ref{hamilton}) around the center we consider
the next change of variables
\begin{equation}\label{cambio}
T:\R^{2}\rightarrow \R^{2}, \quad T(p,q)=\bigg(p-\frac{1}{2},q-\frac{\mu}{\sigma}\bigg).
\end{equation}
We obtain in the new coordinates $(\overline{p},\overline{q})$ the
quadratic Hamiltonian function
\begin{equation}
\mathcal{H}(\overline{p},\overline{q})=-\frac{\sigma}{2}\bigg(\overline{p}^2-\frac{1}{4}\bigg)\bigg(\overline{q}^2-\frac{\mu^2}{\sigma^2}\bigg).
\end{equation}
$\mathcal{H}(\overline{p},\overline{q})$ has the following symmetries
\[
\begin{split}
\text{S}_{1}&: (\overline{p},\overline{q})\rightarrow (\overline{p},-\overline{q}),\\
\text{S}_{2}&: (\overline{p},\overline{q})\rightarrow (-\overline{p},-\overline{q}),\\
\text{S}_{3}&: (\overline{p},\overline{q})\rightarrow (-\overline{p},\overline{q}).
\end{split}
\]
The corresponding dynamical system is
\begin{equation}\label{nps}
\left\{
\begin{aligned}
\dot{\overline{p}}&=-\frac{\sigma}{4}\overline{q}+\sigma \overline{q}\,\overline{p}^2\\
\dot{\overline{q}}&=\frac{\mu^2}{\sigma}\overline{p}-\sigma \overline{p}\,\overline{q}^2
\end{aligned}
\right.
\end{equation}
The invariant lines are now
\[
\left\{\overline{p}=-1/2 \right\}, \quad \left\{\overline{p}=1/2
\right\},\quad \left\{\overline{q}=-\mu/\sigma \right\} \quad
\text{and} \quad \left\{\overline{q}=\mu/ \sigma \right\}.
\]
They define a new quadrangular area $\mathcal{A}$ which center is
our equilibrium point $(0,0)$. See Fig.~\ref{regiona}.

\begin{figure}[h]
\begin{center}
\includegraphics[scale=0.5]{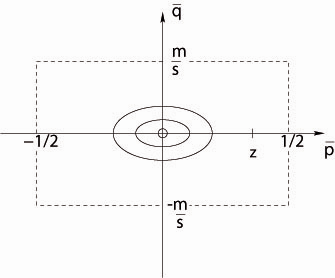}
\caption{Quadrangular region $\mathcal{A}$.}
\label{regiona}
\end{center}
\end{figure}

Consider a periodic solution $\overline{X}(t)=(\overline{p}(t),\overline{q}(t))$ of (\ref{nps}) inside $\mathcal{A}$ with initial conditions
\begin{equation}\label{ic}
\overline{p}(0)=\zeta, \quad \overline{q}(0)=0,
\end{equation}
this implies that $\zeta$ satisfies $-1/2<\zeta<1/2$. Since
$\mathcal{H}$ is a first integral of (\ref{nps}) we take the level
set
\[
h=\mathcal{H}(\zeta,0)=\frac{\mu^2}{2\sigma}(\zeta^2-\frac{1}{4}),
\]
and consider the equation
\[
\begin{split}
\mathcal{H}(\overline{p}(t),\overline{q}(t))&=h\\
-\frac{\sigma}{2}\big(\overline{p}^2(t)-\frac{1}{4}\big)\big(\overline{q}^2(t)-\frac{\mu^2}{\sigma^2}\big)&=\frac{\mu^2}{2\sigma}(\zeta^2-\frac{1}{4}),\\
\end{split}
\]
which can be written as
\[
-(4\overline{p}^2(t)-1)(\sigma^2 \overline{q}^2(t)-\mu^2)=\mu^2(4\zeta^2-1).
\]
By the symmetry $\text{S}_{2}$ we can take
$0<\overline{p}^2(t)<\zeta^2<1/4$  and get from direct calculus
\[
\sigma^2\overline{q}^2(t)=\frac{4\mu^2(\overline{p}^2(t)-\zeta^2)}{4\overline{p}^2(t)-1}.
\]
From (\ref{nps}) it follows that
\[
\overline{q}=\frac{\dot{\overline{p}}}{\sigma(\overline{p}^2-\frac{1}{4})},
\]
and we finally have
\[
\dot{\overline{p}}^2(t)= \frac{1}{4}\mu^2(\overline{p}^{2}(t)-\zeta^2)(4\overline{p}^{2}(t)-1).
\]
Note that the right hand side in the last equation is positive.
Let $T(\zeta)$ be the period of the solution $\overline{X}(t)$ of
(\ref{nps}) satisfying  (\ref{ic}).  We have

\[
\begin{split}
\int_{0}^{T(\zeta)/4}dt&=\int_{0}^{T(\zeta)/4}\frac{-1}{\sqrt{\frac{1}{4}\mu^2(\overline{p}^{2}-\zeta^2)(4\overline{p}^{2}-1)}}\frac{d\overline{p}}{dt}\,dt,\\
T(\zeta)&=\frac{8}{\mu}\int_{0}^{\zeta}\frac{1}{\sqrt{(\overline{p}^{2}-\zeta^2)(4\overline{p}^{2}-1)}}d\overline{p}.
\end{split}
\]
If $\overline{p}=\zeta v$,  then
\[
T(\zeta)=\frac{8}{\mu}K(2\zeta),
\]
where
$\displaystyle{K(x)=\int_{0}^{1}\frac{1}{\sqrt{(1-v^2)(1-x^{2}v^2)}}dv}$
is the complete elliptic integral of the  first kind. On the other
hand the linearized problem of (\ref{nps}) about the equilibrium
solution $(0,0)$ is given by
\begin{equation*}
\left\{
\begin{aligned}
\dot{y}_{1}&=-\frac{\sigma}{4}\,y_{2}\\
\dot{y}_{2}&=\frac{\mu^2}{\sigma}\,y_{1}
\end{aligned}
\right.
\end{equation*}

and its general solution is
\[
\begin{split}
y_{1}&=-C_{1}\frac{\sigma}{2 \mu}\sin \Big(\frac{\mu}{2}t\Big)+C_{2}\cos \Big(\frac{\mu}{2}t\Big),\\
y_{2}&=C_{1}\cos\Big(\frac{\mu}{2}t\Big)+\frac{2\mu}{\sigma}C_{2}\sin \Big(\frac{\mu}{2}t\Big),
\end{split}
\]
This shows that the period of the linear problem is $\displaystyle{T_{lp}=\frac{4\pi}{\mu}}.$

From the above discussions we are now in position to prove the next statement.

\begin{Prop} Let $\overline{X}(t;\zeta)=(\overline{p}(t;\zeta),\overline{q}(t;\zeta))$ be a non trivial periodic solution of (\ref{nps}) that satisfies the initial conditions (\ref{ic}). If $T(\zeta)$ is the period function of $\overline{X}(t;\zeta)$ this function satisfies the following properties
\begin{description}
    \item [(a)] $\displaystyle\lim_{\zeta\to 0^{+}}T(\zeta)=T_{lp}.$
    \item [(b)] $\displaystyle\lim_{\zeta\rightarrow \frac{1}{2}^{-}}T(\zeta)=\infty.$
    \item [(c)] $\displaystyle{\frac{dT}{d\zeta}>0}$ for all $0<\zeta<\frac{1}{2}.$
\end{description}
\end{Prop}

\vspace{0.5 cm}
\noindent
\textbf{Proof.} The proof of $(a)$ and $(b)$ follows from the well know properties of the complete elliptic integral of the first kind function. Moreover
\[
K(2\xi)=\int_{0}^{1}\frac{1}{\sqrt{(1-v^2)(1-4\xi^{2}v^2)}}dv=\int_{0}^{\pi/2}\frac{1}{\sqrt{1-4\xi^2\sin^2\theta}}d\theta.
\]
To prove (c) consider the real function
\[
f(\theta,\zeta)=\frac{1}{\sqrt{1-4\xi^2\sin^2\theta}},
\]

with $f:[0,\pi/2]\times [0,1/2[\to \R$.  This function satisfies
\begin{itemize}
    \item $f(\cdotp,\zeta)$ is a Riemann integrable function in $[0,\pi/2]$, for all $\zeta \in
    [0,1/2[$,
    \item $f(\theta,\cdotp)$ is a differentiable function in $[0,1/2[$, for all $\theta \in [0,\pi/2]$,
\end{itemize}
furthermore $\displaystyle{\frac{\partial f}{\partial \zeta}}$ is
continuous in $[0,\pi/2]\times [0,1/2[$. The previous observations
over the function $f$ are the hypothesis of the classic version of
the rule of derivation under the integral, then
\[
\frac{dT}{d\zeta}=\frac{4\zeta}{\mu}\int_{0}^{\pi/2}\frac{\sin^{2}\theta}{(1-4\zeta^2\sin^2\theta)^{3/2}}d\theta >0,
\]
and this proves (c). $\square$

\vspace{0.5 cm}
\noindent
\textbf{Observation.} By symmetry of the system (\ref{nps}) , the above proposition is also true if we consider  non trivial solutions $\overline{X}(t;\eta)=(\overline{p}(t;\eta),\overline{q}(t;\eta))$ that satisfies the initial conditions
\begin{equation*}
\overline{p}(0)=0, \quad \overline{q}(0)=\eta,
\end{equation*}

\noindent
with $0<\eta <\mu/\sigma$.

\end{document}